\documentclass[conference]{IEEEtran}

\usepackage{multirow}
\usepackage{booktabs}

\usepackage[dvipdfmx]{color}
\usepackage[dvipdfmx]{graphicx}
\usepackage{arydshln}
\usepackage{amsmath,amssymb,amsfonts,mathtools,amsthm,bm,bbm}

\usepackage{cite}

\theoremstyle{definition}

\newtheorem{theorem}{Proposition}
\newtheorem{lemma}{Lemma}
\newtheorem{corollary}{Corollary}

\usepackage{subcaption}

\newcommand{\relmiddle}[1]{\mathrel{}\middle#1\mathrel{}}

\title{Asymptotic Analysis of Normalized SNR-Based Scheduling in Uplink Cellular Networks with Truncated Channel Inversion Power Control}
\author{%
\IEEEauthorblockN{%
\normalsize Shotaro Kamiya\IEEEauthorrefmark{1}\IEEEauthorrefmark{3},
\normalsize Koji Yamamoto\IEEEauthorrefmark{1}\IEEEauthorrefmark{4},
\normalsize Seong-Lyun~Kim\IEEEauthorrefmark{2},
\normalsize Takayuki Nishio\IEEEauthorrefmark{1}, and
\normalsize Masahiro Morikura\IEEEauthorrefmark{1}
}
\IEEEauthorblockA{\small
\IEEEauthorrefmark{1} Graduate School of Informatics, Kyoto University,
Yoshida-honmachi, Sakyo-ku, Kyoto, 606-8501 Japan,\\
\IEEEauthorrefmark{2} The Radio Resource Management and Optimization Laboratory,
School of Electrical and Electronic Engineering,\\ Yonsei University, 50 Yonsei-ro, Seodaemun-gu, Seoul 120-749, Korea
}
\IEEEauthorblockN{\small \IEEEauthorrefmark{3} kamiya@imc.cce.i.kyoto-u.ac.jp, \IEEEauthorrefmark{4} kyamamot@i.kyoto-u.ac.jp
}
}

\begin{document}
\IEEEoverridecommandlockouts
\IEEEpubid{\makebox[\columnwidth]{978-1-5090-3009-5/17/\$31.00~
\copyright2017 IEEE \hfill} \hspace{\columnsep}\makebox[\columnwidth]{ }}
\maketitle

\begin{abstract}
This paper provides the signal-to-interference-plus-noise ratio (SINR) complimentary cumulative distribution function (CCDF) and average data rate
of the normalized SNR-based scheduling in an uplink cellular network using stochastic geometry.
The uplink analysis is essentially different from the downlink analysis in
that the per-user transmit power control is performed and that the interferers are composed of at most one transmitting user
in each cell other than the target cell.
In addition, as the effect of multi-user diversity varies from cell to cell depending on the number of users involved in the scheduling,
the distribution of the number of users is required to obtain the averaged performance of the scheduling.
This paper derives the SINR CCDF relative to the typical scheduled user by focusing on two incompatible cases,
where the scheduler selects a user from all the users in the corresponding Voronoi cell or does not select users near cell edges.
In each case, the SINR CCDF is marginalized over the distribution of the number of users involved in the scheduling,
which is asymptotically correct if the BS density is sufficiently large or small.
Through the simulations, the accuracies of the analytical results are validated for both cases, and the scheduling gains are evaluated
to confirm the multi-user diversity gain.
\end{abstract}

\section{Introduction}

Stochastic geometry is a powerful mathematical and statistical tool,
enabling the tractable modeling of cellular networks without loss of accuracy,
particularly in a multi-cell environment \cite{Andrews2011}.
The success probability $\mathbb{P}(\mathit{SINR}>\theta)$ (equivalent to signal-to-interference-plus-noise ratio (SINR) complimentary cumulative distribution function (CCDF))
can be expressed in a closed-form for a special case
not only in downlink cellular networks \cite{Andrews2011} but also in uplink cellular networks \cite{ElSawy2014}.

The number of papers using some forms of stochastic geometry have increased \cite{Baccelli2009},
implying that the application range of stochastic geometry is growing steadily.
One of the novel applications is the analysis of channel-aware user scheduling \cite{Ohto2017}.
The authors in \cite{Ohto2017} derived the SINR CCDF and the average data rate of the channel-aware scheduling in a downlink cellular network.
They assumed that each transmission in a typical cell follows the normalized SNR-based scheduling \cite{Yang2006}, in which
the resource block in the cell is assigned to a user giving the highest value of the instantaneous SNR normalized by the short-term average SNR.
They enabled the stochastic geometry analysis by considering the largest order statistic to model the fading gain of the scheduled user,
while such an analysis was not conducted in other conventional analyses \cite{Borst2005,Yang2006,Choi2007,Blaszczyszyn2009,Garcia2015}.
As for the relationship between the stochastic geometry analysis and the conventional analyses, please refer to \cite{Ohto2017}.

This paper derives the SINR CCDF and the average data rate
of the typical scheduled user, selected according to the normalized SNR-based scheduling,
in an uplink cellular network with truncated channel inversion power control.
The uplink analysis is essentially different from the typical downlink analysis \cite{Andrews2011} in
that the truncated channel inversion power control per-user is performed \cite{ElSawy2014}
and that the interferers correspond to the set of users composed of at most one scheduled user assigned with the resource block in each cell other than the target cell.
Therefore, we derive the SINR CCDF under the appropriate system model.

In addition, the derived SINR CCDF is marginalized over the distribution of the number of users involved in a particular scheduler.
We would like to point out that the SINR CCDF of the typical scheduled user depends on
the number of users involved in the scheduler because it affects the multi-user diversity gain.
However, the distribution of the number of potential uplink users varies depending on the maximum transmit power, the minimum receiver sensitivity, and the density of BSs.
Note that the maximum transmit power constraint is peculiar to the uplink network, which limits the existing range of the scheduled users to around the corresponding BSs.
This paper considers two distributions of the number of users involved in the scheduling,
for the case where the scheduler selects a user from all the users in the corresponding Voronoi cell
and for the case where the scheduler does not select users near cell edges (see Figs.~\ref{fig:distribution}(\subref{fig:20}) and \ref{fig:distribution}(\subref{fig:10})).
The two distributions are shown to be asymptotically equal to the true distribution.

The contributions of this paper are as follows:
\begin{itemize}
\item
By noticing that the multi-user diversity gain depends on the number of potential uplink users,
we give the SINR CCDF conditioning on the number of potential uplink users
by the sum of Laplace transforms of the probability density function of the aggregate interference.
Note that unlike the downlink analysis \cite{Ohto2017}, some users are not assigned with the resource block due to the maximum transmit power constraint.
For some special cases, the analytical expression reduces to the closed-form expression.
\item
We derive the SINR CCDF and the scheduling gain for the two incompatible cases,
which are asymptotically meaningful in the limit of high/low density of BSs.
Moreover, for each case, we give the probability that indicates the accuracy of the corresponding analytical result
as a function of the BS density.
The analytical results enable the design of uplink cellular networks with taking into account the advantages of channel-aware user scheduling.
\end{itemize}

The rest of this paper is organized as follows.
Section~\ref{sec:system_model} describes the system model.
Section~\ref{sec:derivation} derives the SINR CCDF and the average data rate using stochastic geometry.
Section~\ref{sec:simulations} shows the simulation results, which validate the analytical expression.
Section~\ref{sec:conclusion} concludes this paper.

\section{System Model}\label{sec:system_model}
The locations of BSs and users form independent Poisson point processes (PPPs) in $\mathbb{R}^2$ with respective intensities,
$\lambda_{\text{BS}}$, $\lambda_{\text{UE}}$.
Each user is associated with the nearest BS, meaning that the cell of each BS comprises a Voronoi tessellation on the plane.
We assume that there is one resource block to assign and that each BS serves only one user in the resource block at any given time.

We assume that the desired and interference signals experience path loss with a path loss exponent $\alpha$
and quasi-static Rayleigh fading, i.e., the channel gain is constant over a time slot and is exponentially distributed with mean one.
We assume that the system employs the truncated channel inversion power control \cite{ElSawy2014}.
In this system, each user adjusts the transmit power such that the averaged received signal power at the associated BS is equal to a threshold $\rho_{\mathrm{o}}$,
and the transmit power is limited so as not to exceed a maximum value $P_{\mathrm{u}}$.
Note that some users do experience outage due to insufficient power.
The outage probability was given in \cite{ElSawy2014} by
$\mathrm{e}^{-\pi\lambda_{\text{BS}} (P_{\mathrm{u}} / \rho_{\mathrm{o}})^{2/\alpha}}$,
which means the probability of no BS within $(P_{\mathrm{u}} / \rho_{\mathrm{o}})^{1/\alpha}$,
the distance over which a radio signal transmitted at the maximum transmit power $P_{\mathrm{u}}$ decays to the threshold $\rho_{\mathrm{o}}$ on average, from a typical user.
We should note that the distribution of the number of users involved in a particular scheduling (hereinafter referred to as \textit{involved users})
in a typical cell depends on the outage probability,
and hence on $P_{\mathrm{u}}$, $\rho_{\mathrm{o}}$, and $\lambda_{\text{BS}}$,
which will be discussed in Section~\ref{sec:derivation}.

The system includes a scheduler that selects a user using the resource block.
We consider the normalized SNR-based scheduler \cite{Yang2006}, which assigns the resource block to
the user with the largest instantaneous SNR normalized by the time-averaged SNR of the user over a period when variations induced by fading effects are negligible.
Note that if the data rate is proportional to the SNR \cite{Choi2007},
the normalized SNR-based scheduler is equivalent to the proportional fair scheduler \cite{Viswanath2002}.
A user currently using the resource block in each cell is referred to as an \textit{scheduled user}.
Fig.~\ref{fig:involve_active} illustrates the definitions of involved users and scheduled users to clarify the difference between them.
Note that the authors in \cite{Ohto2017} considered the downlink analysis with the same scheduler,
while this paper discusses the uplink analysis with the truncated channel inversion power control.
We also note that users which experience outage do not belong to the involved users
and that some BSs may not have scheduled users.

Although the system model for uplink cellular networks is based on that of \cite{ElSawy2014} in many aspects,
there are two essential differences in this paper:
we consider the channel-aware user scheduling in which a transmitting user is selected according to the aforementioned manner,
while a transmitting user is randomly chosen in \cite{ElSawy2014};
we take into account the cell including no transmitting user,
while \cite{ElSawy2014} arranges the locations of users such that each BS has at least one scheduled user,
so that the density of interfering users are reduced as a result.

\begin{figure}[t!]
	\centering
  \subcaptionbox{Involved users in typical cell.}{
		\includegraphics[width=0.45\columnwidth]{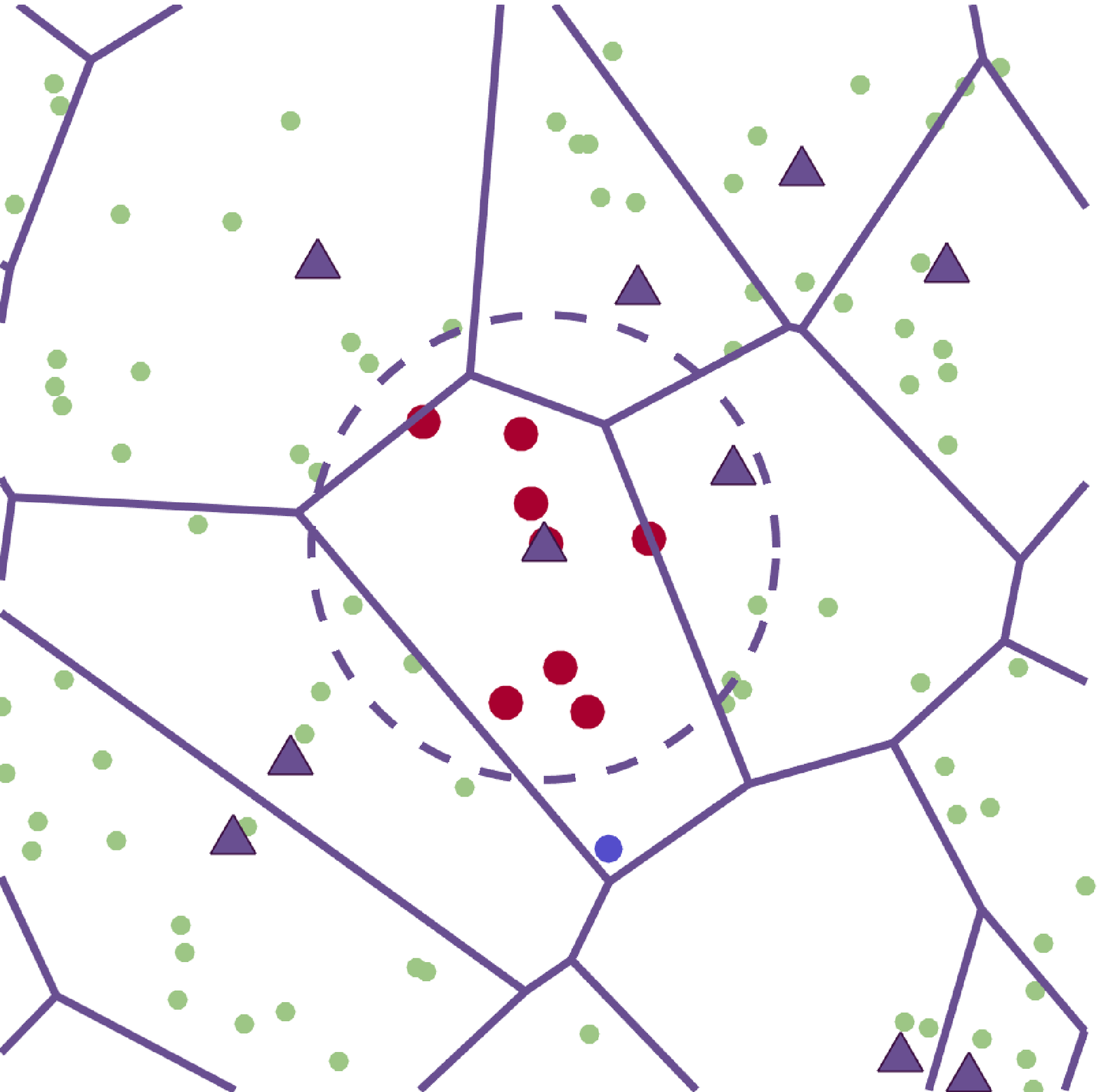}
	}
  \subcaptionbox{Scheduled users.}{
		\includegraphics[width=0.45\columnwidth]{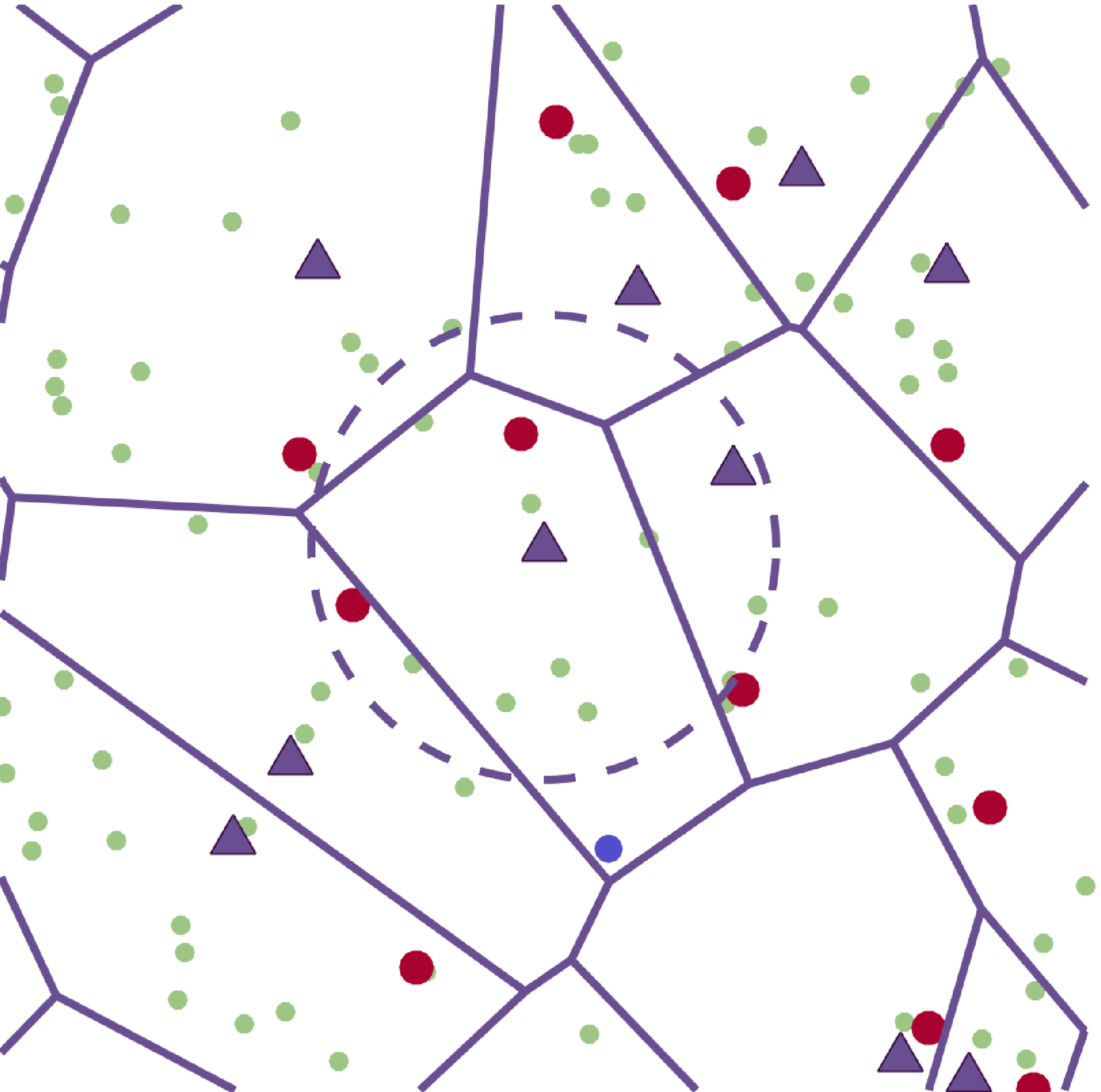}
	}
\caption{Red circles represent each type of users, dashed circle represents achievable range (see Section \ref{ssec:n_dist}),
  purple triangles represent BSs, green dots represent users, blue circles represent outage users in typical cell.}
\label{fig:involve_active}
\end{figure}

\section{SINR CCDF and Average Data Rate}\label{sec:derivation}
The objective in this section is obtaining the SINR CCDF and the average data rate for the described system.
The SINR CCDF of a typical user is the probability that the typical user achieves
some target SINR $\theta$ when the typical user is scheduled, defined as
$\bar{F}_{\mathit{SINR}}(\theta) \coloneqq \mathbb{P}(\mathit{SINR} > \theta)$.

We begin with introducing the SINR with truncated channel inversion power control.
Without loss of generality, we can assume that a typical scheduled user is located at the origin.
The SINR of the typical scheduled user is denoted by
\begin{align}
  \mathit{SINR} = \frac{ \displaystyle \max_{i=1,2,\ldots,n} h_{i} \rho_{\mathrm{o}}}{\sigma^2 + I},
\end{align}
where $h_{i}$ is the fading gain of the involved user $i$ exponentially distributed with mean one,
and $n$ represents the number of involved users in the associated cell.
The max operator reflects the fact that the scheduler selects the user with the largest fading gain.
$\rho_{\mathrm{o}}$ denotes the minimum receiver sensitivity, and $\sigma^2$ is the noise power.
As in \cite{ElSawy2014}, we assume that each user adjusts the transmit power such that the average received signal power at its serving BS is equal to $\rho_{\mathrm{o}}$.
Note that the truncated channel inversion power control is performed in advance of the scheduling,
so that the transmit power of each user does not depend on the user scheduling.
$I$ is the aggregate interference power defined as
\begin{align}
  I = \sum_{u \in \Phi_{\mathrm{iu}}} g_u p_u d_u^{-\alpha},
\end{align}
where $\Phi_{\mathrm{iu}}$ is a point process formed by the locations of interfering users. $g_u$, $p_u$, and $d_u$ represent
the fading gain, the transmit power, and the distance to the typical BS relative to interfering user $u \in \Phi_{\mathrm{iu}}$, respectively.
Note that $p_u,\ \forall u \in \Phi_{\mathrm{iu}}$ is a random variable having the following distribution \cite{ElSawy2014}
\begin{align}
\label{eq:distribution_p}
  f (x;P_{\mathrm{u}}) = \left\{ \begin{array}{ll}
    \frac{2\pi\lambda_{\text{BS}} x^{\frac{2}{\alpha} -1} \mathrm{e}^{-\pi\lambda_{\text{BS}} \left( \frac{x}{\rho_{\mathrm{o}}} \right)^{\frac{2}{\alpha}}} }
    {\alpha \rho_{\mathrm{o}}^{\frac{2}{\alpha}} \left( 1 - \mathrm{e}^{-\pi\lambda_{\text{BS}} \left( \frac{P_{\mathrm{u}}}{\rho_{\mathrm{o}}} \right)^{\frac{2}{\alpha}}} \right)  },&  0\leq x \leq P_{\mathrm{u}};\\
    0, & \text{otherwise,}
    \end{array} \right.
\end{align}
which is obtained by a transformation of variables $p = \rho_{\mathrm{o}} r^{\alpha}$ ($0 \leq p \leq P_{\mathrm{u}}$),
where $r$ is an uplink distance following the Rayleigh distribution with mean $1/2\sqrt{\lambda_{\text{BS}}}$ \cite{Andrews2011,Haenggi2012}.
Note that \cite{ElSawy2014} pointed out that the transmit powers of the scheduled users are identically distributed but are not independent.

In normalized SNR scheduling,
the fading gain of the typical scheduled user follows the distribution of the largest order statistic
for the given fading gains of the involved users in the corresponding cell \cite{Ohto2017}.
The distribution is given by
\begin{align}
\label{eq:cdf_hn}
  \mathbb{P} \left( \max_{i=1,2,\ldots,n} h_{i} \leq x \right) = (1-\mathrm{e}^{-x})^n .
\end{align}
The above cumulative distribution function (CDF) also appears in the calculation of the selection combiner output \cite{Goldsmith2005}.
Note that unlike the selection combining, $n$ is the random variable due to the inherent randomness of Poisson cellular networks.

We now state the main result of this paper.
Before deriving the SINR CCDF, we can give the general result
by exploiting the fact that the multiuser diversity gain increases with the number of involved users associated with the BS.
Letting the number of involved users be a random variable $N$, we have
\begin{align}\label{eq:general}
  \bar{F}_{\mathit{SINR}}(\theta) &= \mathbb{E}_N{\left[\, \mathbb{P}(\,\mathit{SINR} > \theta \mid n\,) \,\right]} \nonumber\\
    &= \sum_{n=0}^\infty f_N(n) \mathbb{P}(\,\mathit{SINR} > \theta \mid n\,),
\end{align}
where $\mathbb{P}(\mathit{SINR} > \theta \mid n)$ denotes the SINR CCDF conditioning on $N=n$,
and $f_N(n)$ is the PMF of the number of involved users in a typical cell.
We first derive $\mathbb{P}(\mathit{SINR} > \theta \mid n)$ (see Lemma \ref{thm:lemma}) as well as some simpler expressions in Section~\ref{ssec:sinr_ccdf_n}.
We then discuss the details of the distribution of the number of involved users in Section~\ref{ssec:n_dist}.

\subsection{SINR CCDF Conditioning on Number of Involved Users}\label{ssec:sinr_ccdf_n}
\begin{lemma}
\label{thm:lemma}
  $\mathbb{P}(\mathit{SINR} > \theta \mid n)$ is given by
	\begin{align}
	\label{eq:lemma}
		&\mathbb{P}(\mathit{SINR} > \theta \mid n) =
		\sum_{k=1}^n \binom{n}{k} (-1)^{k+1} \exp\left( -\frac{k\theta \sigma^2}{\rho_{\mathrm{o}}} \right) \nonumber\\
		&\ \times \exp\left( -  \frac{2 \theta_k \gamma\left(2,\pi \lambda_{\text{BS}} R^2 \right)
		{}_2 F_1 (1,1-\frac{2}{\alpha};2-\frac{2}{\alpha};-\theta_k)}
	  {(\alpha-2)\left( 1 - \mathrm{e}^{-\pi \lambda_{\text{BS}} R^2} \right)/ (1-f_N(0))}\right),
	\end{align}%
	where $R \coloneqq (P_{\mathrm{u}}/\rho_0)^{1/\alpha}$, $\theta_k \coloneqq k\theta$, ${}_2 F_1 (a,b;c;z)$ denotes the Gauss hypergeometric function \cite{Abramowitz1970},
	and $\gamma(a,b)$ denotes the lower incomplete gamma function $\gamma(a,b) \coloneqq \int_0^b t^{a-1} \mathrm{e}^{-t}\mathrm{d}t$.
\end{lemma}
\begin{proof}
We have
\begin{align}
\label{eq:ccdf_conditioned_n}
  &\mathbb{P}(\mathit{SINR} > \theta \mid n) = \mathbb{E}_I{\left[\, \mathbb{P} \left( \max_{i=1,2,\ldots,n} h_{i}  > \frac{\theta(\sigma^2 + I)}{\rho_{\mathrm{o}}}  \right) \relmiddle| I \,\right]} \nonumber\\
    &= \mathbb{E}_{I}{\left[1- \left(1- \exp\left( - \frac{\theta(\sigma^2 + I)}{\rho_{\mathrm{o}}}\right) \right)^n \right]} \nonumber\\
    &= \sum_{k=1}^n \binom{n}{k} (-1)^{k+1} \exp\left( -\frac{k\theta \sigma^2}{\rho_{\mathrm{o}}} \right) \mathbb{E}_I{\left[\exp\left( -\frac{k\theta I}{\rho_{\mathrm{o}}} \right)\right]} \nonumber\\
    &= \sum_{k=1}^n \binom{n}{k} (-1)^{k+1} \exp\left( -\frac{k\theta \sigma^2}{\rho_{\mathrm{o}}} \right) \mathcal{L}_I{\left( \frac{k\theta}{\rho_{\mathrm{o}}} \right)},
\end{align}
where $\mathcal{L}_I(s)$ denotes the Laplace transform of the probability density function
of the aggregate interference power.
Note that $\mathbb{P}(\mathit{SINR} > \theta \mid n)$ is given based on the expectation of the sum of exponential functions of interference $I$
because the CDF of the fading gain of the scheduled user \eqref{eq:cdf_hn} can be written as a sum of exponential functions from the binomial theorem.
Although this property is true for the downlink analysis \cite{Ohto2017}, the part regarding the Laplace transform is different
due to the power control and the difference in interference sources.

We can derive $\mathcal{L}_I (k\theta / \rho_{\mathrm{o}})$ in the same way in \cite{ElSawy2014}, i.e., by using two approximations:
the scheduled users form a PPP and their transmit powers are independent.
Note that these approximations are made for the analytical tractability \cite{ElSawy2014}.
Considering that the scheduled users constitute a homogeneous PPP with intensity $\lambda^{(n>0)}_{\text{BS}} \coloneqq (1-f_N(0)) \lambda_{\text{BS}}$,
and that their transmit powers are independent random values, we obtain
\begin{align}
\label{eq:laplace1}
  \mathcal{L}_I(s) = \exp\left( - 2 \pi \lambda^{(n>0)}_{\text{BS}} s^{2/\alpha} \mathbb{E}_p\left[ p^{2/\alpha}\right] \int^\infty_{(s\rho_{\mathrm{o}})^{\frac{-1}{\alpha}}} \frac{y}{y^\alpha + 1} \,\mathrm{d}y \right),
\end{align}
where $\mathbb{E}_p [ p^{2/\alpha} ]$ is obtained in \cite{ElSawy2014} with the distribution of transmit power \eqref{eq:distribution_p} as
\begin{align}
  \mathbb{E}_p\left[ p^{2/\alpha}\right] = \frac{\rho_{\mathrm{o}}^{2/\alpha} \gamma\left(2,\pi \lambda_{\text{BS}} R^2 \right)}{ \pi \lambda_{\text{BS}} \left(1 - \mathrm{e}^{-\pi \lambda_{\text{BS}} R^2 }\right) }.
\end{align}
While the accuracy of these approximations is validated from \cite{ElSawy2014} for the case without channel-aware scheduling,
it is also validated in Section~\ref{sec:simulations} even if the channel-aware scheduling is employed.
It should be noted that compared to the existing analysis \cite{ElSawy2014}, the density of interfering users that appears
when applying the probability generating functional
reduces to $(1-f_N(0))\lambda_{\text{BS}}$ due to the absence of users in some cells.

Substituting $s = k\theta /\rho_{\mathrm{o}}$ into \eqref{eq:laplace1} yields
\begin{align}
\label{eq:laplace2}
  &\mathcal{L}_I \left( \frac{k\theta}{\rho_{\mathrm{o}}} \right) =
    \exp\left( -  \frac{2 \theta_k^{2/\alpha} \gamma\left(2,\pi \lambda_{\text{BS}} R^2 \right) \int^\infty_{\theta_k^{\frac{-1}{\alpha}}} \frac{y}{y^\alpha + 1} \,\mathrm{d}y }
    { \left(1 - \mathrm{e}^{-\pi \lambda_{\text{BS}} R^2} \right) / (1-f_N(0)) } \right) \nonumber \\
  &=\exp\left( -  \frac{2 \theta_k \gamma\left(2,\pi \lambda_{\text{BS}} R^2 \right) {}_2 F_1 (1,1-\frac{2}{\alpha};2-\frac{2}{\alpha};-\theta_k)}
    {(\alpha-2)\left( 1 - \mathrm{e}^{-\pi \lambda_{\text{BS}} R^2} \right) / (1-f_N(0))}\right).
\end{align}
Finally, substituting \eqref{eq:laplace2} into \eqref{eq:ccdf_conditioned_n} yields \eqref{eq:lemma}.
\end{proof}

Simpler expressions can be derived for some special cases.
For the interference-limited case, i.e., $\sigma^2 + I \simeq I$, $\exp( -k\theta \sigma^2/\rho_{\mathrm{o}}) \mathcal{L}_I (k\theta/\rho_{\mathrm{o}})$
reduces to $\mathcal{L}_I (k\theta/\rho_{\mathrm{o}})$.
When $\alpha = 4$, the integral part reduces to the closed-form expression, $\arctan(\sqrt{k\theta})/2$.
In the case of $P_{\mathrm{u}} \to \infty$ (equivalently $R \to \infty$), $\gamma(2,\pi\lambda_{\text{BS}}R^2)$ reduces to $1$ and $1 - \mathrm{e}^{-\pi\lambda_{\text{BS}}R^2} \to 1$.
The simplest form is obtained when the above three conditions are simultaneously satisfied, which leads to
\begin{align}
  &\mathbb{P}(\mathit{SINR}>\theta \mid n) \nonumber \\
	&\ = \sum_{k=1}^{n} \binom{n}{k} (-1)^{k+1} \exp\left( -  (1-f_N(0)) \sqrt{\theta_k} \arctan \sqrt{\theta_k} \right).
\end{align}
Note that the above expression is a closed form and only depends on a threshold $\theta$,
and substituting $n=1$ yields the SINR CCDF without the channel-aware scheduling, which was obtained in \cite{ElSawy2014}.

\subsection{Distribution of the Number of Involved Users}\label{ssec:n_dist}
The derivation of the SINR CCDF is complete if we have the distribution of the number of involved users in a typical cell, $f_N(n)$.
In fact, the distribution depends on the existing range of a scheduled user.
To discuss the existing range, we define the achievable region of the BS
as the circle centered in the typical BS with radius $R=(P_{\mathrm{u}}/\rho_{\mathrm{o}})^{1/\alpha}$.
A user outside the achievable range of the serving BS
is not admitted to transmit frames due to the maximum power constraint.
In this paper, we give the distribution for the two incompatible cases,
where the achievable range includes the Voronoi cell and where the Voronoi cell includes the achievable range,
shown as Fig.~\ref{fig:distribution}(\subref{fig:20}) and Fig.~\ref{fig:distribution}(\subref{fig:10}), respectively.

\begin{figure}[t!]
	\centering
	\subcaptionbox{Voronoi cell. \label{fig:20}}{
		\includegraphics[width=0.29\columnwidth]{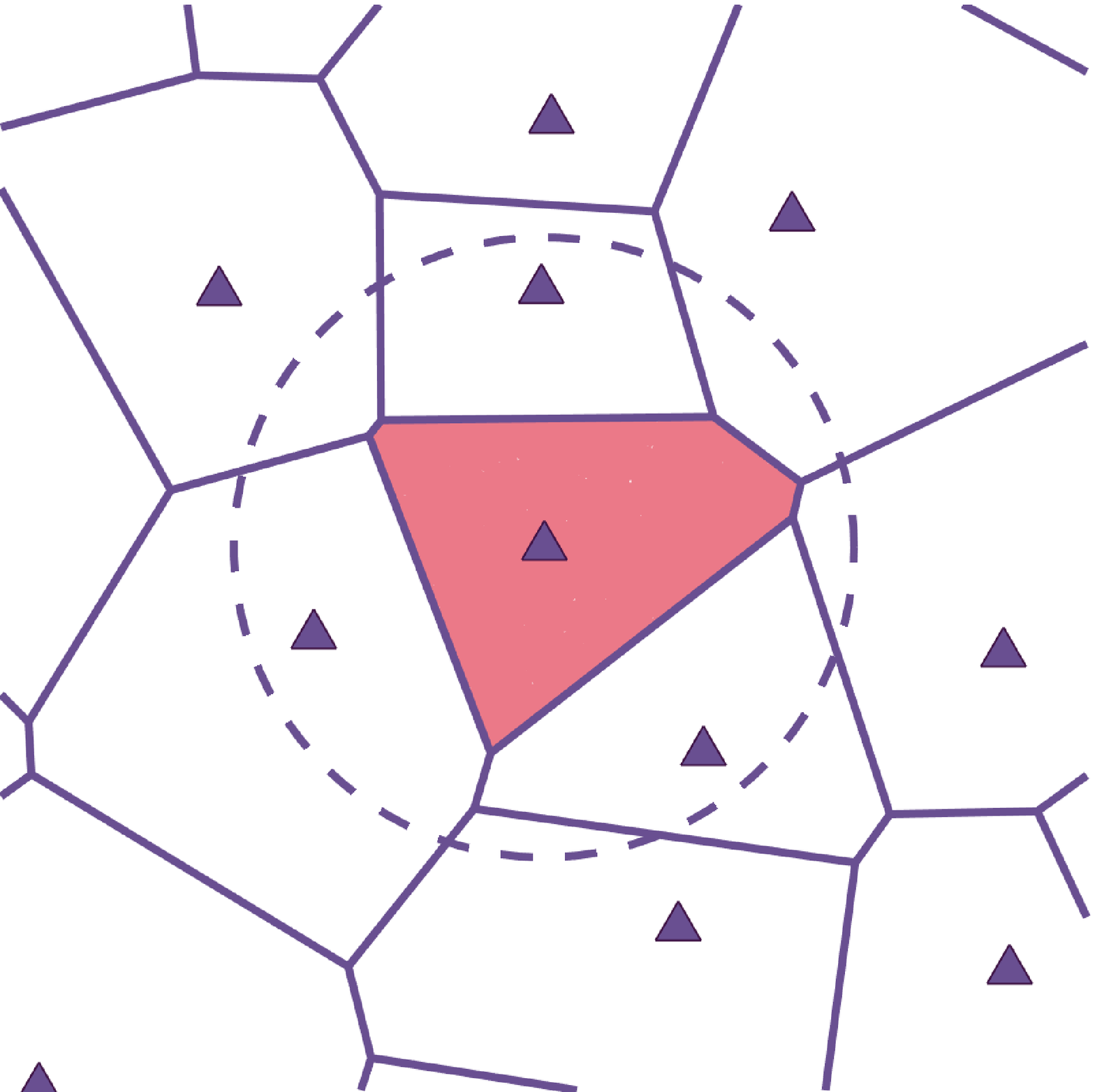}
	}
	\subcaptionbox{Partial overlap. \label{fig:15}}{
		\includegraphics[width=0.29\columnwidth]{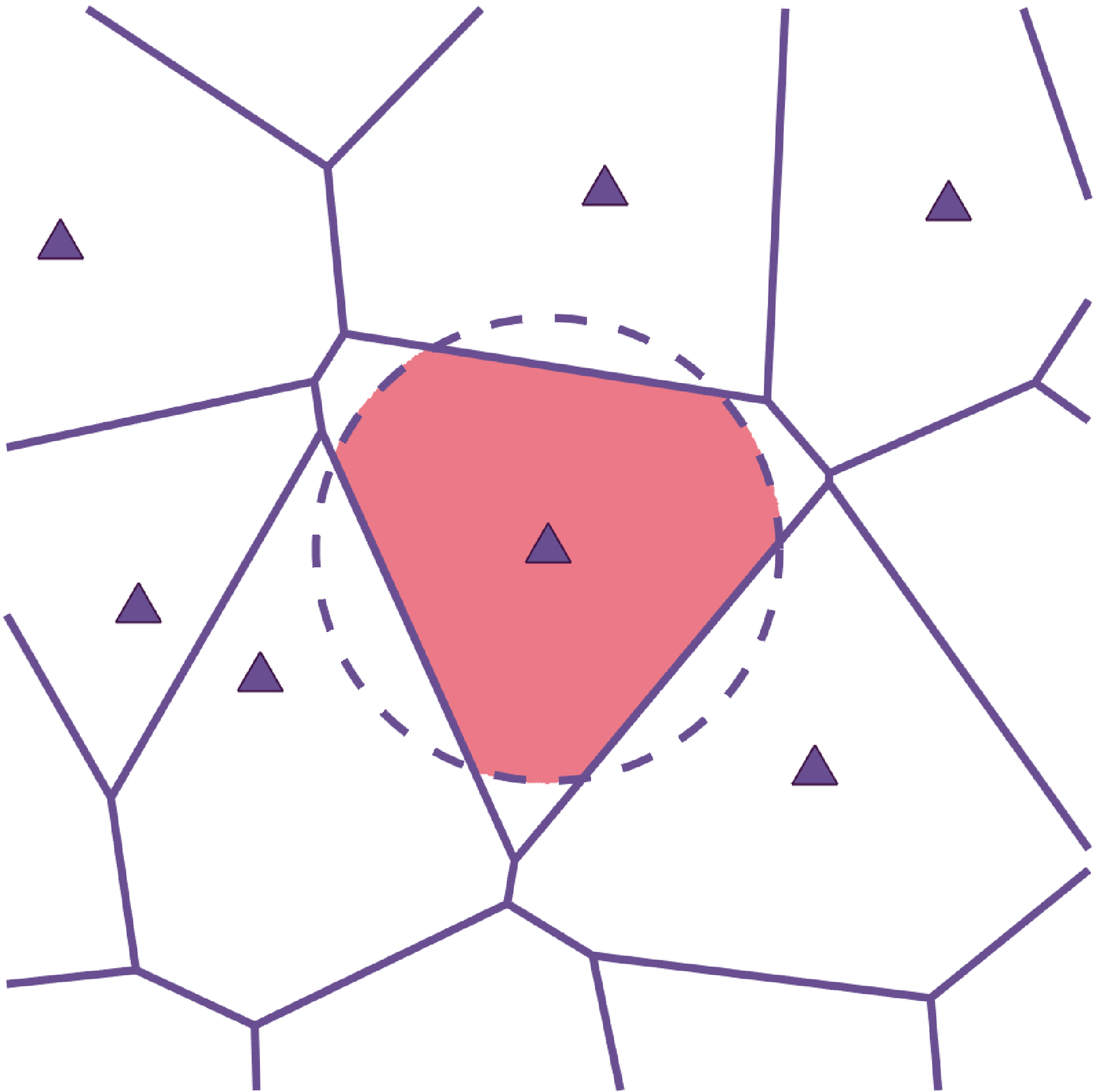}
	}
  \subcaptionbox{Achievable range. \label{fig:10}}{
		\includegraphics[width=0.29\columnwidth]{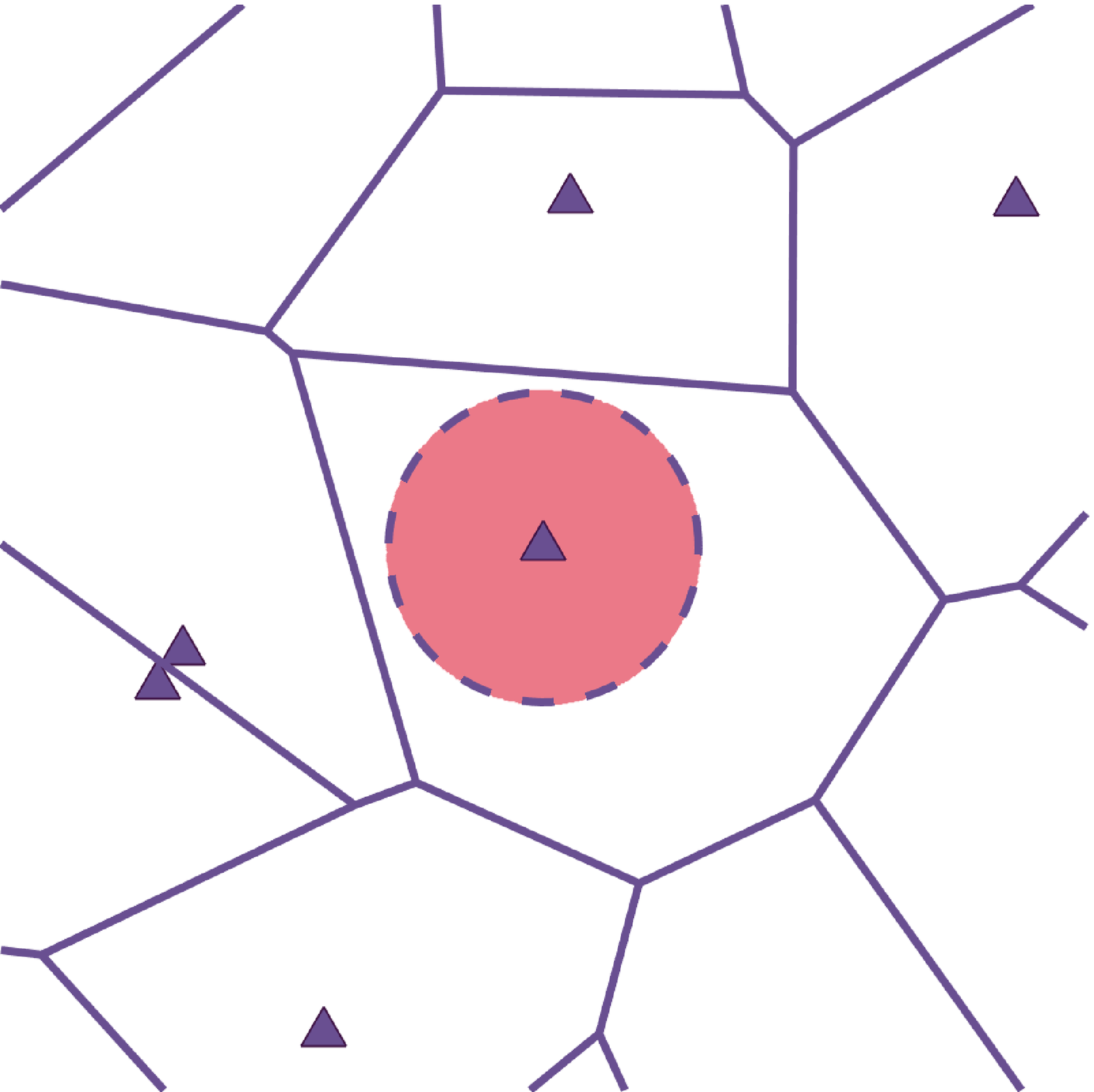}
	}
\caption{Existing range of scheduled user; red region represents existing range, dashed circle represents achievable range,
  purple triangles represent BSs.}
\label{fig:distribution}
\end{figure}
\begin{figure}[t!]
	\centering
	\includegraphics[width=0.95\columnwidth]{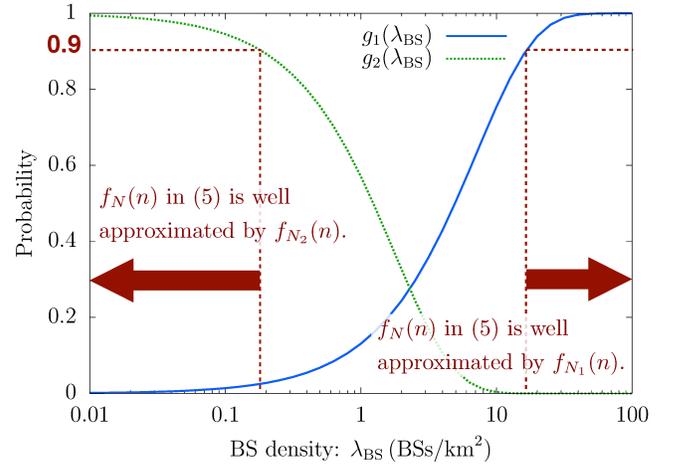}
	\caption{Probabilities indicating validity of distributions
	when $P_{\mathrm{u}} = 23\,\text{dBm}$, $\rho_{\mathrm{o}} = -70\,\text{dBm}$, and $\alpha =4$.}
	\label{fig:distribution_range}
\end{figure}

\setcounter{equation}{15}
\begin{figure*}[!t]
\vspace{-7mm}
\begin{align}
\label{eq:sinr_ccdf}
	\bar{F}_{\mathit{SINR}}(\theta) = \sum_{n=0}^\infty f_{N_i}(n) \sum_{k=1}^n \binom{n}{k} (-1)^{k+1} \exp\left( -\frac{k\theta \sigma^2}{\rho_{\mathrm{o}}}
		-  \frac{2 \theta_k \gamma\left(2,\pi \lambda_{\text{BS}} R^2 \right) {}_2 F_1 (1,1-\frac{2}{\alpha};2-\frac{2}{\alpha};-\theta_k)}
			{(\alpha-2)\left( 1 - \mathrm{e}^{-\pi \lambda_{\text{BS}} R^2} \right) / (1-f_{N_i}(0))}\right),\ i \in \{1,2\}
\end{align}
\vspace{-2mm}
\hrulefill
\end{figure*}
\setcounter{equation}{11}

In the case where the achievable range includes the Voronoi cell,
i.e., where a scheduled user exists anywhere in the Voronoi cell,
the distribution of the number of involved users is approximately given by
\begin{align}
\label{eq:distribution1}
  f_{N_1}(n) = \frac{(\lambda_{\text{UE}} / c \lambda_{\text{BS}})^n}{ c B(n+1,c-1) (\lambda_{\text{UE}} / c\lambda_{\text{BS}} + 1)^{n+c}},
\end{align}
which is the PMF of the number of users in a typical Voronoi cell \cite{Yu2013},
where $B(\alpha,\beta)$ is the beta function and $c=3.5$.
The probability that a typical Voronoi cell is included in the corresponding achievable range,
denoted by $g_{1} (\lambda_{\text{BS}})$, is equivalent to the probability that at least one BS exists within $R$ from a typical user,
yielding
\begin{align}\label{eq:g1}
	g_{1}(\lambda_{\text{BS}}) = 1 - \mathrm{e}^{-\pi\lambda_{\text{BS}} R^2}.
\end{align}
The distribution of the number of the involved users in the typical cell $f_N(n)$ in \eqref{eq:general} is well approximated by $f_{N_1}(n)$
if $\lambda_{\text{BS}}$ is sufficiently large
so that $g_{1}(\lambda_{\text{BS}}) \simeq 1$.

In the case where the Voronoi cell includes the achievable range,
i.e., where users at cell edges are not scheduled,
the distribution of the number of involved users is given by
\begin{align}
\label{eq:distribution2}
  f_{N_2}(n)= \frac{(\lambda_{\text{UE}} \pi R^2)^n}{n!} \mathrm{e}^{- \lambda_{\text{UE}} \pi R^2},
\end{align}
which is the PMF of the number of users forming the PPP with intensity $\lambda_{\text{UE}}$
in the range with area $\pi R^2$, obtained from the definition of the PPP \cite{Haenggi2012}.
The probability that a typical Voronoi cell includes the corresponding achievable range, denoted by $g_2(\lambda_{\text{BS}})$,
is equivalent to the probability that no cell edges exist within $R$ from the BS.
We can find that the probability is equivalent to the probability that no other BS exists within the circle centered at the typical BS with radius $2R$,
yielding
\begin{align}\label{eq:g2}
	g_2(\lambda_{\text{BS}}) = \mathrm{e}^{-4\pi\lambda_{\text{BS}} R^2}.
\end{align}
$f_N(n)$ in \eqref{eq:general} is well approximated by $f_{N_2}(n)$
if $\lambda_{\text{BS}}$ is sufficiently small so that $g_2(\lambda_{\text{BS}}) \simeq 1$.

Two probabilities indicating the validity of the distributions are shown in Fig.~\ref{fig:distribution_range}.
We show the range of the BS density in which either probability exceeds $0.9$, which means
more than 90\% of the Voronoi cells includes or is included by the achievable range,
and therefore that either $f_{N_1}(n)$ or $f_{N_2}(n)$ would approach the true distribution $f_N(n)$.
In the simulation part, we show that the Monte Carlo simulation result is well approximated by the analytical result
for the BS density at which the corresponding probability shown in Fig.~\ref{fig:distribution_range} is approximately equal to $0.9$.
Obtaining the distribution for the case where the achievable range and the Voronoi cell partially overlap, shown as Fig.~\ref{fig:distribution}(\subref{fig:15}),
will be our future work.

\subsection{Complete SINR CCDF and Average Data Rate}
Substituting \eqref{eq:ccdf_conditioned_n} and \eqref{eq:distribution1} or \eqref{eq:distribution2} into \eqref{eq:general},
we obtain \eqref{eq:sinr_ccdf} and the following proposition.
\begin{theorem}
  The SINR CCDF $\bar{F}_{\mathit{SINR}}(\theta)$ of the normalized SNR-based scheduling in uplink cellular networks is given by \eqref{eq:sinr_ccdf},
	which is appropriate if $g_i(\lambda_{\text{BS}}) \simeq 1$ ($i=1,2$).
\end{theorem}

Using the SINR CCDF, we can obtain the average data rate $\tau_{s}(\lambda_{\text{BS}},\lambda_{\text{UE}}) \coloneqq \mathbb{E}[\ln(1+\mathit{SINR})]$
and, therefore, the scheduling gain \cite{Berggren2004} $G(\lambda_{\text{BS}},\lambda_{\text{UE}}) \coloneqq \tau_{s}(\lambda_{\text{BS}},\lambda_{\text{UE}})/ \tau_{r}(\lambda_{\text{BS}})$, where
$\tau_{r}(\lambda_{\text{BS}})$ is the average data rate of round-robin scheduling given as
\setcounter{equation}{16}
\begin{align}
\label{eq:rate_rr}
  \tau_{r}(\lambda_{\text{BS}}) = (1-f_N(0)) \int_0^\infty \frac{\mathrm{e}^{-\frac{x \sigma^2}{\rho_{\mathrm{o}}}}}{x+1} \mathcal{L}_I \left( \frac{x}{\rho_{\mathrm{o}}} \right) \mathrm{d}x,
\end{align}
where the rate is assumed to be zero for cells where there are no scheduled users.
Note that \eqref{eq:rate_rr} is different from that of \cite{ElSawy2014}
in that we consider BSs having no users to serve.

\begin{corollary}
  The average data rate of the normalized SNR-based scheduling is given by \eqref{eq:rate}.
\end{corollary}
\begin{proof}
We have
\begin{align}
\label{eq:rate}
  &\mathbb{E}[\ln(1+\mathit{SINR})] = \int_0^\infty \mathbb{P}(\ln(1+\mathit{SINR})> t) \,\mathrm{d}t \nonumber\\
	&=\int_0^\infty \mathbb{P}(\mathit{SINR} > \mathrm{e}^t - 1) \,\mathrm{d}t = \int_0^\infty \bar{F}_{\mathit{SINR}}(\mathrm{e}^t - 1) \,\mathrm{d}t \nonumber\\
  &= \int_0^\infty \mathbb{E}_N{\left[ \sum_{k=1}^n \binom{n}{k} (-1)^{k+1} \right.}\nonumber\\
  &\ \ \ \ \ \ \left. \times  \exp\left( -\frac{k(\mathrm{e}^t-1) \sigma^2}{\rho_{\mathrm{o}}} \right) \mathcal{L}_I \left( \frac{k(\mathrm{e}^t-1)}{\rho_{\mathrm{o}}} \right) \mathrm{d}t \right]\nonumber\\
  &= \mathbb{E}_N{\left[ \sum_{k=1}^n \binom{n}{k} (-1)^{k+1} \right.}\nonumber\\
  &\ \ \ \ \ \ \left. \times \int_0^\infty \exp\left( -\frac{k(\mathrm{e}^t-1) \sigma^2}{\rho_{\mathrm{o}}} \right) \mathcal{L}_I \left( \frac{k(\mathrm{e}^t-1)}{\rho_{\mathrm{o}}} \right) \mathrm{d}t \right]\nonumber\\
  &= \mathbb{E}_N{\left[ \sum_{k=1}^n \binom{n}{k} (-1)^{k+1} \int_0^\infty \frac{\mathrm{e}^{-\frac{kx \sigma^2}{\rho_{\mathrm{o}}}}}{x+1} \mathcal{L}_I \left( \frac{kx}{\rho_{\mathrm{o}}} \right) \mathrm{d}x \right]}.
\end{align}
\end{proof}

\begin{figure}[t!]
	\centering
  \subcaptionbox{$\lambda_{\text{BS}}=0.2\,\text{BSs/km}^2$. \label{fig:sim1}}{
		\includegraphics[width=0.95\columnwidth]{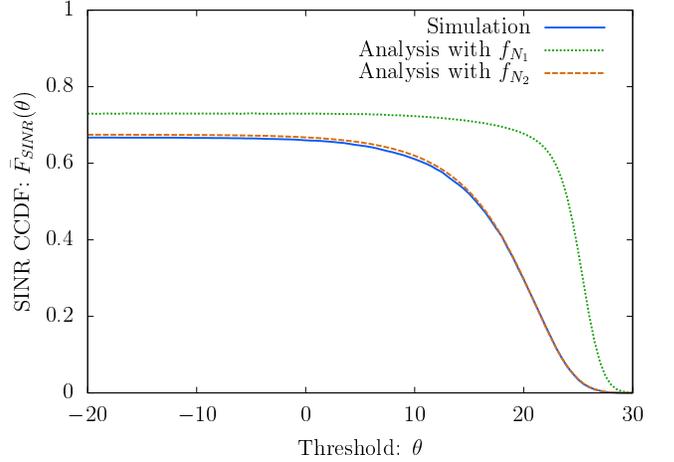}
	}
	\subcaptionbox{$\lambda_{\text{BS}}=2\,\text{BSs/km}^2$. \label{fig:sim3}}{
		\includegraphics[width=0.95\columnwidth]{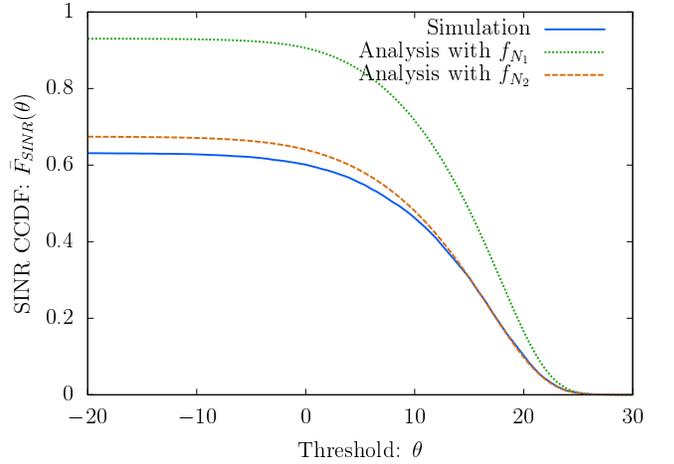}
	}
	\subcaptionbox{$\lambda_{\text{BS}}=20\,\text{BSs/km}^2$. \label{fig:sim2}}{
		\includegraphics[width=0.95\columnwidth]{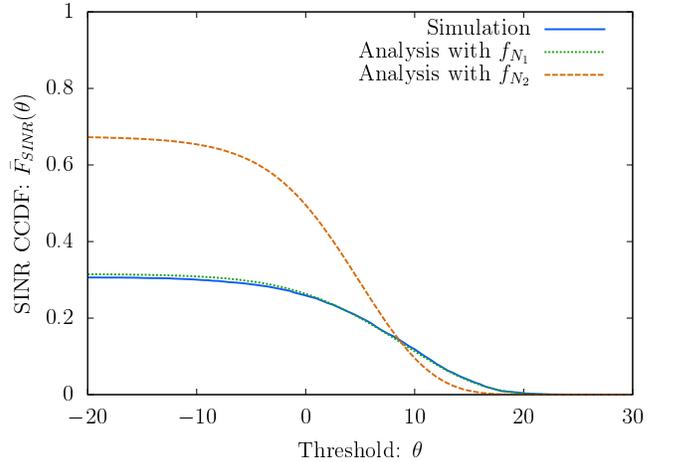}
	}
\caption{SINR CCDF $\bar{F}_{\mathit{SINR}}(\theta)$.}
\label{fig:sims}
\end{figure}

Letting $\sigma^2=0$, $P_{\mathrm{u}} \to \infty$, $\alpha=4$ yields the simplest form:
\begin{align}
  &\mathbb{E}[\ln(1+\mathit{SINR})] = \nonumber \\
  &\ \mathbb{E}_N{\left[ \sum_{k=1}^{n} \binom{n}{k} (-1)^{k+1} \int_0^\infty
    \frac{ \mathrm{e}^{- (1-f_{N_i}(0)) \sqrt{kx} \arctan(\sqrt{kx}) } \,\mathrm{d}x }{x+1} \right]}.
\end{align}

\section{Validation Through Simulations}\label{sec:simulations}
To validate the analytical result, we investigate the SINR CCDF through simulations.
Unless otherwise specified, we set 
the user density $\lambda_{\text{UE}}=8\,\text{BSs/km}^2$, the path loss exponent $\alpha=4$,
the noise power $\sigma^2 = -90\,\text{dBm}$, the maximum transmit power $P_{\mathrm{u}} = 23\,\text{dBm}$,
and the minimum receiver sensitivity of BSs $\rho_{\mathrm{o}} = -70\,\text{dBm}$.
Each simulation is repeated 10,000 times, and the corresponding result shows the average value.

Figs.~\ref{fig:sims}(\subref{fig:sim1}), \ref{fig:sims}(\subref{fig:sim3}), and \ref{fig:sims}(\subref{fig:sim2}) show the SINR CCDFs
for the cases of $\lambda_{\text{BS}}=0.2,\ 2,\ 20\,\text{BSs/km}^2$, respectively.
In all the cases, we also show two analytical results \eqref{eq:sinr_ccdf} for $i=1,2$.
We can see that the simulation results coincide with the analytical result averaged over $f_{N_2}(n)$ for the case of $\lambda_{\text{BS}}=0.2\,\text{BSs/km}^2$
and with the result averaged over $f_{N_1}(n)$ for the case of $\lambda_{\text{BS}}=20\,\text{BSs/km}^2$.
We note that $g_1(20\,\text{BSs/km}^2) = 0.940$ and $g_2(0.2\,\text{BSs/km}^2) = 0.894$,
so that the analytical result with $f_{N_1}(n)$ and that with $f_{N_2}(n)$ are valid when
$\lambda_{\text{BS}}=20\,\text{BSs/km}^2$ and $\lambda_{\text{BS}}=0.2\,\text{BSs/km}^2$, respectively.
On the other hand, for the case of $\lambda_{\text{BS}}=2\,\text{BSs/km}^2$, neither of the two analytical results coincides with the simulation result.
Obtaining $g_1(2\,\text{BSs/km}^2) = 0.245$ and $g_2(2\,\text{BSs/km}^2) = 0.325$ implies the analytical result using either $f_{N_1}(n)$ or $f_{N_2}(n)$ is not valid
for this case.

Note that the reason why the supremum is less than one is
\begin{align*}
  \MoveEqLeft \lim_{\theta \to -\infty} \mathbb{P}(\mathit{SINR}>\theta) = \lim_{\theta \to -\infty} \sum_{n=1}^\infty f_{N_i} (n) \mathbb{P}(\mathit{SINR}>\theta \mid n)\\
    &= \sum_{n=1}^\infty f_{N_i}(n) = 1 - f_{N_i}(0).
\end{align*}
The supremum approaches to one if we set the densities $\lambda_{\text{BS}}$ and/or $\lambda_{\text{UE}}$ such that $f_{N}(0) \simeq 0$
according to \eqref{eq:distribution1} or \eqref{eq:distribution2}.

\begin{figure}[t!]
	\centering
	\includegraphics[width=0.95\columnwidth]{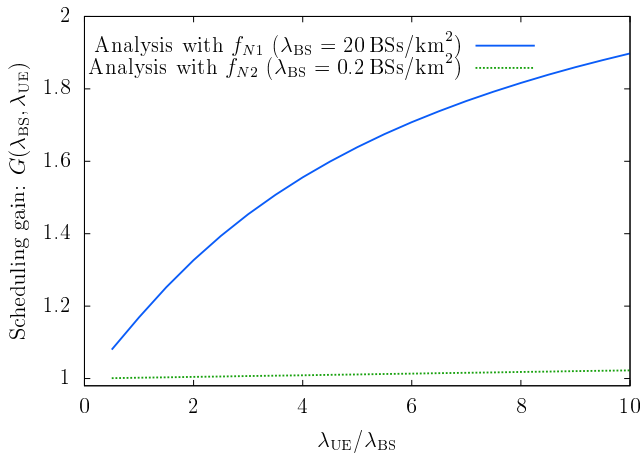}
	\caption{Scheduling gain $G(\lambda_{\text{BS}},\lambda_{\text{UE}})$.}
	\label{fig:scheduling_gain}
\end{figure}

In addition, Fig.~\ref{fig:scheduling_gain} shows the scheduling gains for both cases $\lambda_{\text{BS}}=0.2,\ 20\,\text{BSs/km}^2$.
We can see that the scheduling gains increase along with the ratio $\lambda_{\text{UE}}/\lambda_{\text{BS}}$ for both cases.
Note that the scheduling gain for $\lambda_{\text{BS}}=0.2\,\text{BSs/km}^2$ remains small in the given range $\lambda_{\text{UE}}/\lambda_{\text{BS}} \leq 10$
because the number of achievable users is still small due to relatively small $\lambda_{\text{UE}}$.

\section{Conclusions}\label{sec:conclusion}
In this paper, we derived the SINR CCDF and the scheduling gain of the normalized SNR-based scheduling
in uplink cellular networks.
The SINR CCDF and the average data rate were obtained in the form of the expected value
in terms of the number of users involved in a particular scheduling, which affects the multi-user diversity.
Noticing that the maximum transmit power constraint restricts the existing range of a scheduled user unlike the downlink analysis,
we provided the distributions of the number of users for two incompatible cases:
where all the users in the corresponding cell are scheduled and where users at cell edges are not scheduled.
We confirmed the accuracy of the analysis through Monte Carlo simulations for both cases,
and we observed that the scheduling gain increased as the involved users per cell increased owing to the multi-user diversity,
as was also observed in the downlink analysis \cite{Ohto2017}.

\section*{Acknowledgment}
This work was supported in part by JSPS KAKENHI Grant Number JP17J04854 and KDDI Foundation.

\bibliographystyle{IEEEtran}
\bibliography{my_bib.bib}

\end{document}